\documentclass[fleqn,12pt,twoside]{article}
\usepackage{tipa}

\usepackage[headings]{espcrc1}
\readRCS
$Id: espcrc1.tex,v 1.2 2004/02/24 11:22:11 spepping Exp $
\usepackage{graphicx}
\usepackage[figuresright]{rotating}

\newtheorem{theorem}{Theorem}

\newtheorem{corollary}{Corollary}

\newtheorem{lemma}{Lemma}

\newtheorem{proposition}{Proposition}

\newenvironment{proof}[1][Proof.]{\begin{trivlist}
\item[\hskip \labelsep {\bfseries #1}]}{\end{trivlist}}

\newcommand{\AmS}{{\protect\the\textfont2
  A\kern-.1667em\lower.5ex\hbox{M}\kern-.125emS}}

\usepackage{amsmath}
\usepackage{amsfonts}
\title{On strongly spanning $k$-edge-colorable subgraphs}

\author{Vahan V. Mkrtchyan\address[MCSD]{Department of Informatics and Applied Mathematics,\\
Yerevan State University, Yerevan, 0025, Armenia}%
\address{Institute for Informatics and Automation Problems,\\
National Academy of Sciences of Republic of Armenia, 0014, Armenia}
\thanks{email: vahanmkrtchyan2002@\{ysu.am, ipia.sci.am,
yahoo.com\}},        
                and
        Gagik N. Vardanyan\addressmark[MCSD]\thanks{email: vgagik@gmail.com.}}

\runtitle{On strongly spanning $k$-edge-colorable subgraphs} \runauthor{Vahan
Mkrtchyan, Gagik Vardanyan}

\begin{document}

\maketitle

\begin{abstract}
A subgraph $H$ of a multigraph $G$ is called strongly spanning, if any vertex of $G$ is not isolated in $H$, while it is called maximum $k$-edge-colorable, if $H$ is proper $k$-edge-colorable and has the largest size. We introduce a graph-parameter $sp(G)$, that coincides with the smallest $k$ that a graph $G$ has a strongly spanning maximum $k$-edge-colorable subgraph. Our first result offers some alternative definitions of $sp(G)$. Next, we show that $\Delta(G)$ is an upper bound for $sp(G)$, and then we characterize the class of graphs $G$ that satisfy $sp(G)=\Delta(G)$. Finally, we prove some bounds for $sp(G)$ that involve well-known graph-theoretic parameters. 
\end{abstract} \\ 



\section{Introduction}

Let $N$ denote the set of positive integers. In this paper we consider multigraphs. They are assumed to be finite, undirected and without loops, though they may contain multiple edges. If $G$ is a multigraph, then for a vertex $x\in V(G)$ $d_G(x)$ denotes the degree of $x$ in $G$. Moreover, let $\Delta (G)$ and $\delta (G)$ denote the maximum and minimum degrees of vertices in $G$, respectively. A vertex is defined to be isolated in $G$, if its degree is zero. If $G^{\prime }$ is a subgraph of $G$, then we say that $G^{\prime }$ covers (misses) a vertex $x$ of $G$, if $d_{G^{\prime }}(x)\geq 1$ ($d_{G^{\prime }}(x)=0$). A subgraph is strongly spanning, if it covers all the vertices of the graph. A point that should be made clear here, is that if a vertex $x$ of $G$ is not a vertex of a subgraph $G^{\prime }$, then we assume that $d_{G^{\prime }}(x)=0$.

The length of a path $P$ of a multigraph $G$ is the number of edges lying on $P$. If $a,b$ are non-negative integers, then a subgraph $H$ of a multigraph $G$ with $V(H)=V(G)$ is called an $[a,b]$-factor of $G$ if for any vertex $v$ of $G$ $a\leq d_H(v) \leq b$. A subset $E'$ of edges of a multigraph $G$ is called matching, if $(V(G),E')$ is a $[0,1]$-factor of $G$. Clearly, matchings can be defined as a set of edges that contain no adjacent edges. Usually, a vertex that is (not) incident to an edge from a matching, is said to be covered (missed) by the matching. A matching is maximum, if it has the largest cardinality. A matching is perfect, if any vertex is incident to an edge from the matching.

A proper $k$-edge-coloring of a multigraph $G$ is an assignment of colors from a set of $k$ colors such that adjacent edges receive different colors. Observe that a proper $k$-edge-coloring of a multigraph $G$ can be viewed a partition of $E(G)$ into $k$ matchings. Usually, these matchings into which $E(G)$ is partitioned, are called color-classes of the edge-coloring.  The least integer $k$ for which $G$ has a proper $k$-edge-coloring is called the chromatic index of $G$ and is denoted by $\chi^{\prime}(G)$. Clearly, $\chi^{\prime}(G)\geq \Delta(G)$ for any multigraph $G$, and the following classical theorems of Shannon and Vizing give non-trivial upper bounds for $\chi^{\prime}(G)$:

\begin{theorem}
(Shannon \cite{Shannon}). For every multigraph $G$%
\begin{equation}
\Delta(G)\leq \chi ^{\prime }(G)\leq \left[ \frac{3\Delta(G)}{2}\right].
\end{equation}
\end{theorem}

\begin{theorem}
(Vizing, \cite{Vizing}). For every multigraph $G$ 
\begin{equation*}
\Delta(G)\leq \chi ^{\prime }(G)\leq \Delta(G) +\mu(G),
\end{equation*} where $\mu(G)$ denotes the maximum multiplicity of an edge in $G$.
\end{theorem}

Note that Shannon's theorem implies that if we consider a cubic
multigraph $G$, then $3\leq \chi ^{\prime }(G)\leq 4$, thus $\chi ^{\prime
}(G)$ can take
only two values. In 1981 Holyer proved that the problem of deciding whether $%
\chi ^{\prime }(G)=3$ or not for cubic multigraphs $G$ is NP-complete \cite%
{Holyer}, thus the calculation of $\chi ^{\prime }(G)$ is already hard
for cubic multigraphs.

For a multigraph $G$ and $k\in N$, let 
\begin{equation*}
\nu_k(G)=\{|E(H_k)|:H_k\textrm{ is a proper }k\textrm{-edge-colorable subgraph of }G\}.
\end{equation*} A proper $k$-edge-colorable subgraph of $G$ containing $\nu_k(G)$ edges will be called a maximum $k$-edge-colorable subgraph. We define $\nu(G)=\nu_1(G)$.

The quantitative aspect of the investigation of maximum $k$-edge-colorable subgraphs
of multigraphs and particularly, $r$-regular multigraphs has attracted a lot of attention, previously. The basic problem that researchers were interested was the following: what is
the proportion of edges of a multigraph (or an $r$-regular multigraph, and particularly, cubic
multigraph), that we can cover by its $k$ matchings? 

For the case $k=1$ in \cite{Hobbs} an investigation is carried out in the class of cubic graphs, and in
\cite{Bella,HenningYeo,TakaoBaybars,Takao,Weinstein} for the
general case. Let us also note that the relation between $\nu _{1}(G)$ and $%
\left\vert V\right\vert $ has also been investigated in the regular
multigraphs of high girth \cite{GirthBound}.

The same is true for the case $k=2,3$. Albertson and Haas
investigate these ratios in the class of cubic and $4$-regular graphs in \cite
{AlbertsonHaasFirst,AlbertsonHaasSecond}, and Steffen investigates
the problem in the class of bridgeless cubic multigraphs in \cite{Steffen}. Similar investigations are done in \cite{Rizzi_2009} for subcubic multigraphs. In \cite{part1} the problem is addressed in the class of cubic multigraphs. Finally, a best-possible bound is proved in \cite{DeltaSubgraphs} for the case $k=\Delta(G)$ in the class of all multigraphs.

However, it worths to be mentioned that the quantitative line of the research was not the only one. Previously, a special attention was also paid to structural properties of maximum $k$-edge-colorable subgraphs, and sometimes this kind of results have helped researchers to get quantitative results. A typical example of a structural result is the one proved in \cite{AlbertsonHaasSecond}, which states that in any cubic multigraph $G$ there is a maximum $2$-edge-colorable subgraph $H$, such that the multigraph $G\backslash E(H)$ is $2$-edge-colorable. Recently, in \cite{DeltaSubgraphs} new such results are presented for maximum $\Delta(G)$-edge-colorable subgraphs of multigraphs $G$. In particular, it is shown there that any
set of vertex-disjoint cycles of a multigraph $G$ (particularly, any $2$-factor) can be extended to a maximum $\Delta(G)$-edge-colorable subgraph of $G$ if $\Delta(G)\geq 3$. Also, it is shown there that for any maximum $\Delta(G)$-edge-colorable
subgraph $H$ of $G$ $|\partial_{H}(X)|\geq\lceil\frac{|\partial_{G}(X)|}{2}\rceil$ for each $X\subseteq V(G)$, where $\partial_K(X)$ is the set of edges of a multigraph $K$ with exactly one end-vertex in $X$. Finally, in \cite{part2} it is shown that the edges of a cubic multigraph lying outside a maximum $3$-edge-colorable subgraph form a matching. Though this result does not have a direct generalization, using the ideas of the proof of Vizing theorem for graphs from \cite{West}, in \cite{DeltaSubgraphs} it is shown that a graph $G$ has a maximum $\Delta(G)$-edge-colorable subgraph $H$, such that the edges of $G$ that do not belong to $H$ form a matching.

In this paper, we concentrate on strongly spanning maximum $k$-edge-colorable subgraphs of multigraphs. In the beginning of the paper we introduce a graph-parameter $sp(G)$, that coincides with the smallest $k$ that a graph $G$ has a strongly spanning maximum $k$-edge-colorable subgraph. We first give some alternative definitions of $sp(G)$. Then, we show that $\Delta(G)$ is an upper bound for $sp(G)$, and we proceed with the characterization of graphs $G$ with $sp(G)=\Delta(G)$. Finally, we relate $sp(G)$ to some well-known graph-theoretic parameters.

Non-defined terms and concepts can be found in \cite{Lov,West}.

\section{The main results}

We start with a lemma, that will allow us to look at our main parameter from various perspectives.

\begin{lemma}\label{ContinuousSpectra} If a multigraph $G$ has a strongly spanning $k$-edge-colorable subgraph, then it has a strongly spanning maximum $k$-edge-colorable subgraph.
\end{lemma}

\begin{proof} Let $A_k$ be a strongly spanning $k$-edge-colorable subgraph. Consider all maximum $k$-edge-colorable subgraphs of $G$, and among them choose the ones that cover maximum possible number of vertices. From these subgraphs, choose a subgraph $H_k$ such that $|E(A_k)\cap E(H_k)|$ is maximized. Let us show that $H_k$ is a strongly spanning subgraph.

On the opposite assumption, consider a vertex $u$ missed by $H_k$. Consider the vertices $u_1,...,u_q$ ($q\geq 1$) that are adjacent to $u$. Since $H_k$ is a maximum $k$-edge-colorable subgraph of $G$, we have:

\begin{enumerate}
	\item [(a)] $d_G(u_i)\geq k+1$ for $i=1,...,q$;
	\item [(b)] $d_{H_k}(u_i)= k$ for $i=1,...,q$.
\end{enumerate} Let $v_i$ be any neighbour of the vertex $u_i$ ($1\leq i \leq q$) with $d_{H_k}(v_i)\geq 1$. Note that (a) implies that such a vertex $v_i$ exists, moreover, it  is different from $u$. Let us show that
\begin{enumerate}
	\item [(c)] $d_{H_k}(v_i)=1$.
\end{enumerate} Now if $d_{H_k}(v_i)\geq 2$, then define a subgraph $H^{\prime}_k$ of $G$ as follows:
\begin{equation*}
H^{\prime}_k=(H_k\backslash \{(u_i,v_i)\})\cup \{(u,u_i)\}.
\end{equation*}Clearly $H^{\prime}_k$ is a maximum $k$-edge-colorable subgraph of $G$. Moreover, $H^{\prime}_k$ covers more vertices of $G$ than $H_k$ does, which contradicts the choice of $H_k$. Thus (c) must hold.

We are ready to complete the proof of the lemma. Since $A_k$ is a strongly spanning $k$-edge-colorable subgraph, there is an edge $e=(u,w)\in E(A_k)$. By (b), we have $d_{H_k}(w)= k$, thus there is an edge $f=(w,z)\in E(H_k)$ such that $f\notin E(A_k)$. Consider a subgraph $H^{\prime \prime}_k$ of $G$ as follows:
\begin{equation*}
H^{\prime \prime}_k=(H_k\backslash \{f\})\cup \{e\}.
\end{equation*}Clearly $H^{\prime \prime}_k$ is a maximum $k$-edge-colorable subgraph of $G$. Due to (c), $H^{\prime \prime}_k$ covers maximum possible number of vertices, like $H_k$. However,
\begin{equation*}
|E(A_k)\cap E(H^{\prime \prime}_k)|>|E(A_k)\cap E(H_k)|,
\end{equation*} which contradicts the choice of $H_k$. The proof of Lemma \ref{ContinuousSpectra} is completed. $\square$

\end{proof}
 
Next, we prove the following theorem.

\begin{theorem}\label{characthm} For $k\in N$ and a multigraph $G$ without isolated vertices, the following conditions are equivalent:
\begin{enumerate}
\item [(a)] $G$ contains a $[1,k]$-factor,
\item [(b)] $G$ contains a strongly spanning $k$-edge-colorable subgraph,
\item [(c)] $G$ contains a strongly spanning maximum $k$-edge-colorable subgraph.
\end{enumerate}
\end{theorem}
\begin{proof} Since a maximum $k$-edge-colorable subgraph is a $k$-edge-colorable subgraph, (c) implies (b). Moreover, since a strongly spanning $k$-edge-colorable subgraph is a $[1,k]$-factor, (b) implies (a). By Lemma \ref{ContinuousSpectra}, we already have that (b) implies (c). Thus, it suffices to show that (a) implies (b). 

Let $H$ be a $[1,k]$-factor of $G$. Let $T$ be a sub-forest of $H$ with $V(T)=V(H)=V(G)$. Clearly, $T$ is a strongly spanning subgraph of $G$. Since $T$ is $\Delta(T)$-edge-colorable and  $\Delta(T)\leq \Delta(H)\leq k$, we have $T$ is $k$-edge-colorable. Hence (a) implies (b). The proof of Theorem \ref{characthm} is completed.
$\square$
\end{proof}

\begin{corollary} If a multigraph has a perfect matching, then it has a strongly spanning maximum $k$-edge-colorable subgraph for all values of $k$.
\end{corollary}

We are ready to introduce our main parameter. If $G$ is a multigraph without isolated vertices, then define:
\begin{equation*}
sp(G)=\min\{k:G \textrm{ has a strongly spanning maximum $k$-edge-colorable subgraph}\}.
\end{equation*} Observe that due to Theorem \ref{characthm}, $sp(G)$ coincides with the least $k$ such that $G$ has a strongly spanning $k$-edge-colorable subgraph. Similarly, $sp(G)$ represents the smallest $k$ for which $G$ has a $[1,k]$-factor.

A multigraph $G$ without isolated vertices can be viewed as a $[1,\Delta(G)]$-factor of $G$, thus we have: 
\begin{equation}\label{trivialupperbound}
1\leq sp(G)\leq \Delta(G).
\end{equation}

The following theorem of Tutte characterizes multigraphs $G$ with $sp(G)=1$.

\begin{theorem}(Tutte, see Theorem 3.1.1 from \cite{Lov}) A multigraph $G$ has a perfect matching, if and only if for any $S\subseteq V(G)$ one has $o(G-S)\leq |S|$, where for a multigraph $H$ $o(H)$ denotes the number of components of $H$ that contain odd number of vertices.
\end{theorem} We will also need the Tutte-Berge formula, which can be shown to be equivalent to the mentioned theorem of Tutte (see Theorem 3.1.14 from \cite{Lov}).

\begin{theorem}(Tutte-Berge formula) For any multigraph $G$ 
\begin{equation*}
\max_{S\subseteq V(G)}(o(G-S)- |S|)=|V(G)|-2\nu(G).
\end{equation*}
\end{theorem}

Now, let us characterize the class of multigraphs with $sp(G)=\Delta(G)$. Clearly, if $G_1,...,G_t$ are components of $G$, then $sp(G)=\max\{sp(G_1),...,sp(G_t)\}$. Thus, a multigraph $G$ satisfies the equality $sp(G)=\Delta(G)$ if and only if some of its components satisfies the same equality. This observation enables us to focus on the characterization of connected multigraphs $G$ that satisfy $sp(G)=\Delta(G)$.

\begin{lemma}\label{spDeltaTreeCycle} If $G$ is a connected multigraph with $sp(G)=\Delta(G)$, then either $G$ is an odd cycle or $G$ is a tree.
\end{lemma}
\begin{proof}Let $G$ be a counter-example to this statement minimizing $|E(G)|$. Let us show that $G$ is unicyclic, that is, $G$ contains exactly one cycle.

Since $G$ is not a tree, it must contain a cycle. Let us assume that $G$ contains at least two cycles, and let $e$ be an edge of $G$ lying on a cycle of $G$. Observe that:
\begin{equation*}
sp(G)\leq sp(G-e)\leq \Delta(G-e) \leq \Delta(G).
\end{equation*}Taking into account that $sp(G)=\Delta(G)$, we have that $sp(G-e)=\Delta(G-e)$. Since $G-e$ is connected and $|E(G-e)|=|E(G)|-1<|E(G)|$, we have that $G-e$ is either a tree or an odd cycle. Now, if $G-e$ is a tree, then $G$ must be unicyclic \cite{West}, which we assumed to be not the case. Hence $G-e$ is an odd cycle. However, this case is also impossible since if $G-e$ is an odd cycle, then $\Delta(G)=3$ and $sp(G)\leq 2$, and therefore $sp(G)<\Delta(G)$, which contradicts the choice of $G$. We conclude that $G$ is unicyclic.

Let $C$ be the cycle of $G$. Observe that since $G$ is not a cycle ($G\neq C$), it must contain a vertex of degree one. 

Let us show that any degree one vertex of $G$ is adjacent to a vertex of $C$. On the opposite assumption, we can consider a vertex $u$ of $G$ such that $d_G(u)=p+1\geq 2$ and $u$ is adjacent to $p\geq 1$ vertices of degree one. Let $u_1,...,u_p$ be the degree one neighbours of $u$, and let $v$ be the other neighbour of $u$. Observe that since $G$ is not a tree, $v$ is not of degree one. Let $G_1$ be the component of $G-(u,v)$ containing the vertex $v$. Clearly, $C$ is a cycle of $G_1$. We need to consider two cases.\\

Case 1: $G_1=C$. In this case, we have that $\Delta(G)=\max\{d_G(v),d_G(u)\}=\max\{3,p+1\}$ and $sp(G)\leq \max\{2,p\}$, hence $sp(G)<\Delta(G)$, which contradicts the choice of $G$.\\

Case 2: $G_1\neq C$. Since $G_1$ is connected, $G_1$ contains a cycle and $|E(G_1)|<|E(G)|$, we have that $sp(G_1)\leq \Delta(G_1)-1<\Delta(G)$. Hence $sp(G)\leq \max\{sp(G_1),p\}<\Delta(G)$, since $\Delta(G)\geq p+1$, which contradicts the choice of $G$.\\

The considered two cases imply that any degree one vertex of $G$ is adjacent to a vertex of $C$. Observe that this implies that all vertices of $G$ that are of degree at least two, lie on $C$. We are ready to complete the proof of the lemma. For this purpose we consider the following two cases, and in each of them we exhibit a contradiction.\\

Case 1: $G$ contains two degree two vertices that are adjacent. Let $u$ and $v$ be adjacent degree two vertices of $G$, and let $u_1$ and $v_1$ be the other ($\neq v$ and $\neq u$) neighbours of $u$ and $v$, respectively. Consider the multigraph $G'$ obtained from $G$ by removing the vertices $u$ and $v$, and adding an edge connecting $u_1$ and $v_1$.  Since $G'$ is connected and $|E(G')|<|E(G)|$, we have that $sp(G')\leq \Delta(G')-1=\Delta(G)-1$. Let $H'$ be a strongly spanning $(\Delta(G)-1)$-edge-colorable subgraph of $G'$. Consider a subgraph $H$ of $G$ obtained from $H'$ as follows:
\begin{equation*}
H=\left\{
\begin{array}{ll}
(H'\backslash \{(u_1,v_1)\})\cup \{(u,u_1), (v,v_1)\}, & \text{if } (u_1,v_1)\in E(H');\\
 H'\cup \{(u,v)\}, & \text{if } (u_1,v_1)\notin E(H').
\end{array}
\right.
\end{equation*}It is easy to see that $H$ is a strongly spanning $(\Delta(G)-1)$-edge-colorable subgraph of $G$, hence $sp(G)\leq \Delta(G)-1$ contradicting the choice of $G$.\\

Case 2: $G$ contains no two degree two vertices that are adjacent. Observe that this case includes the case when there are no degree two vertices in $G$. For each degree two vertex $u$ of $G$ choose the edge $(u,u')$ incident to $u$ such that $u'$ is the next neighbour of $u$ in the direction of clockwise circumvention of $C$, and let $M$ be the matching of $G$ that contains all such edges $(u,u')$. Consider a subgraph $H$ of $G$ obtained as follows: all edges of $G$ that are incident to a degree one vertex add to $H$, and add $M$ to $H$, too. Clearly, $H$ is a strongly spanning $(\Delta(G)-1)$-edge-colorable subgraph of $G$, hence $sp(G)\leq \Delta(G)-1$ contradicting the choice of $G$.\\

The proof of Lemma \ref{spDeltaTreeCycle} is completed.
$\square$
\end{proof}

Lemma \ref{spDeltaTreeCycle} implies that in order to characterize the connected multigraphs $G$ with $sp(G)= \Delta(G)$, we can focus on trees. For this purpose, for an arbitrary tree $T$, we introduce the following two sets:
\begin{equation*}
A=\{v\in V(T):d_T(v)=\Delta(T)\}, B=V(T)\backslash A.
\end{equation*}

\begin{lemma}\label{treeDeltaMinus1} Let $T$ be a tree with $|V(T)|\geq 3$. Then for any $v\in B$ there is a $(\Delta(T)-1)$-edge-colorable subgraph $H$ of $G$, such that either $V(H)=V(T)$ or $V(T)\backslash V(H)=\{v\}$.
\end{lemma}
\begin{proof}We will give a method for the construction of such a subgraph. We start with $H=\emptyset$. Consider the following partition of vertices of $T$: 
\begin{equation*}
V_0=\{v\}, V_1=\{u:(v,u)\in E(T)\},...,V_p=\{u:(z,u)\in E(T)\textrm{ and } z\in V_{p-1}\}.
\end{equation*}Now, add all edges $(z,u)$ to $H$, such that $u\in V_p$ and $z\in V_{p-1}$.  Observe that for any $w\in V(H)\cap V_{p-1}$ one has $d_H(w)\leq \Delta(T)-1$ since $w$ has one neighbour in $V_{p-2}$. After this, remove all edges that we have added to $H$ and the vertices incident to them from $T$. Repeat this process until $V(T)$ becomes empty or $V(T)=\{v\}$.

It can be easily seen that the components of the resulting subgraph $H$ of $T$ are stars, such that their centers are of degree at most $\Delta(T)-1$. Hence $H$ is $(\Delta(T)-1)$-edge-colorable. Moreover, it meets the requirements of the lemma.
$\square$
\end{proof}

In the following two corollaries, for a tree $T$, $H$ denotes the subgraph from Lemma \ref{treeDeltaMinus1}. 

\begin{corollary}\label{spdeltacor}If $T$ is a tree with $|E(T)|\geq 3$ and $sp(T)=\Delta(T)$, then $V(T)\backslash V(H)=\{v\}$.
\end{corollary}

\begin{corollary}\label{spdeltacor2}If $T$ is a tree with $|E(T)|\geq 3$ and the subgraph $H$ does not cover $v$, then there is a strongly spanning $\Delta(T)$-edge-colorable subgraph $H'$ of $T$, such that $d_{H'}(v)=1$.
\end{corollary}

Now, we introduce an operation that will help us to characterize the trees $T$ with $sp(T)=\Delta(T)$. Let $T_1$ be a tree with $|V(T_1)|\geq 3$, and let $K_{1,p}$ be a star with $p\geq 2$. Consider the tree $T=T_1 \circ K_{1,p}$ obtained from $T_1$ and $K_{1,p}$ by identifying a degree one vertex of $K_{1,p}$ with a vertex $v\in B=B(T_1)$. First, we establish some properties of the operation $\circ$.

\begin{lemma}\label{OperProp} Let $T_1$ be a tree with $|V(T_1)|\geq 3$, and let $K_{1,p}$ be a star with $p\geq 2$. If $T=T_1 \circ K_{1,p}$ then:
\begin{enumerate}
	\item [(a)]if $p< sp(T_1)=\Delta(T_1)$, then $sp(T)\neq \Delta(T)$;
	\item [(b)]if $p\leq sp(T_1)<\Delta(T_1)$, then $sp(T)\neq \Delta(T)$;
	\item [(c)]if $sp(T_1)<p$, then $sp(T)\neq \Delta(T)$;
	\item [(d)]if $p=sp(T_1)=\Delta(T_1)$, then $sp(T)=\Delta(T)$.
\end{enumerate}
\end{lemma}
\begin{proof}Let $L=\max\{\Delta(T_1),p\}$. Clearly, $\Delta(T)=L$. Suppose that the tree $T$ has been obtained from $T_1$ and $K_{1,p}$, by identifying the vertices $w\in B=B(T_1)$, and the degree one vertex $u\in V(K_{1,p})$. Moreover, let $z$ be the center of $K_{1,p}$.

(a) Since $\Delta(T_1)>p$, then $\Delta(T)=\Delta(T_1)$. Let us show that $sp(T)\leq \Delta(T)-1$. As $w\in B=B(T_1)$, Corollary \ref{spdeltacor} implies that there is a $(\Delta(T_1)-1)$-edge-colorable subgraph $H_1$ of $T_1$, such that $V(T_1)\backslash V(H_1)=\{w\}$. Consider the subgraph $H$ of $T$ obtained from $H_1$ by adding $E(K_{1,p})$ to it. Clearly, $H$ is $(\Delta(T)-1)$-edge-colorable subgraph of $T$, hence $sp(T)\leq \Delta(T)-1 < \Delta(T)$.

(b) Clearly, $\Delta(T)=\Delta(T_1)$. Let us show that $sp(T)\leq sp(T_1)<\Delta(T)$. Take a strongly spanning $sp(T_1)$-edge-colorable subgraph $H_1$ of $T_1$. Consider the subgraph $H$ of $T$ obtained from $H_1$ by adding $E(K_{1,p})\backslash \{(u,z)\}$ to it. Clearly, $H$ is a strongly spanning $sp(T_1)$-edge-colorable subgraph of $T$. Hence $sp(T)\leq sp(T_1)$.

(c) Let us show that $sp(T)\leq p-1<\Delta(T)$. Take a strongly spanning $sp(T_1)$-edge-colorable subgraph $H_1$ of $T_1$. Consider the subgraph $H$ of $T$ obtained from $H_1$ by adding $E(K_{1,p})\backslash \{(u,z)\}$ to it. Clearly, $H$ is a strongly spanning $(p-1)$-edge-colorable subgraph of $T$. Hence $sp(T)\leq p-1$.

(d) Clearly, $\Delta(T)=\Delta(T_1)=p$. Suppose that $k=sp(T)<\Delta(T)=p$, and let $H$ be a strongly spanning $k$-edge-colorable subgraph of $T$. Set: $H_1=H\cap E(T_1)$. 

Observe that $(w,z)\notin E(H)$, as otherwise $E(K_{1,p})\subseteq E(H)$ and hence all edges of $K_{1,p}$ would have to be colored, which would mean that $k=p$. This implies that $H_1$ is a strongly spanning $k$-edge-colorable subgraph of $T_1$, hence $sp(T_1)\leq k<p=\Delta(T_1)$, which contradicts our assumption.
$\square$
\end{proof}

We are ready to characterize the trees $T$ with $sp(T)=\Delta(T)$. For that purpose, for any two trees $T'$ and $T''$, we write $T' \rightarrow T''$, if $T''$ can be obtained from $T'$ by the application of Lemma \ref{OperProp}(d).

\begin{theorem}\label{TreeCharactthm} A tree $T$ satisfies $sp(T)=\Delta(T)$, if and only if, there is a sequence of trees $T_0, T_1,..., T_m$ ($m\geq 0$), such that $T_0$ is a star, $T_m=T$, $sp(T_j)=\Delta(T_j)$ for $j=0,1,...,m$ and $T_0 \rightarrow T_1 \rightarrow \ldots \rightarrow T_m$.
\end{theorem}
\begin{proof} If $T$ is a star, then clearly $sp(T)=\Delta(T)$. On the other hand, if $T$ is obtained from a star $T_0$ by applying Lemma \ref{OperProp}(d), then by Lemma \ref{OperProp}(d), all intermediate trees $T_j$ satisfy $sp(T_j)=\Delta(T_j)$. Hence $sp(T)=\Delta(T)$.

Now, assume that $T$ satisfies $sp(T)=\Delta(T)$. Let us show the existence of the corresponding sequence of trees. If $T$ is a star, we are done. Otherwise, assume that $T$ is not a star. Then, there is a vertex $z$ of $T$, that is of degree $p\geq 2$, such that $z$ is adjacent to exactly $p-1$ vertices of degree one. Let $T'$ be the tree obtained from $T$ by removing the vertex $z$ and all its neighbours that are of degree one. Moreover, let $w$ be the vertex of $T'$ such that $(z,w)\in E(T)$. Let us show that $T=T' \circ K_{1,p}$. 

Clearly, it suffices to show that $w\in B=B(T')$. Suppose that $w\in A=A(T')$, that is $d_{T'}(w)=\Delta(T')$. Then, clearly, $\Delta(T)=\max\{d_T(w),d_T(z)\}=\max\{\Delta(T')+1,p\}$. Consider a strongly spanning subgraph $H$ of $T$ obtained from any strongly spanning $\Delta(T')$-edge-colorable subgraph of $T'$ by adding all edges incident to $z$ except $(z,w)$. It is not hard to see that $H$ is $\max\{\Delta(T'),p-1\}$-edge-colorable, hence $sp(T)\leq \max\{\Delta(T'),p-1\}<\max\{\Delta(T')+1,p\}=\Delta(T)$ contradicting the choice of $T$.

Lemma \ref{OperProp} implies that $T'$ and $p$ satisfy the conditions of Lemma \ref{OperProp}(d). Hence, $T' \rightarrow T$. By induction, there is a sequence of trees $T_0, T_1,..., T_m$ ($m\geq 0$), such that $T_0$ is a star, $T_m=T'$, $sp(T_j)=\Delta(T_j)$ for $j=0,1,...,m$ and $T_0 \rightarrow T_1 \rightarrow \ldots \rightarrow T_m$. Consider the sequence of trees $T_0, T_1,..., T_m, T_{m+1}$, where $T_{m+1}=T$. Observe that it meets the requirements of the theorem. The proof of Theorem \ref{TreeCharactthm} is completed.
$\square$
\end{proof}

Now we turn to the problem of finding some bounds for $sp(G)$ in terms of well-known graph theoretic parameters.

Thomassen has shown that any almost regular multigraph $G$ (that is, a multigraph $G$ with $\Delta(G)-\delta(G)\leq 1$) has a $[1,2]$-factor \cite{Thom}, hence we have:

\begin{proposition}\label{almostsp} For any almost regular multigraph $G$ $sp(G)\leq 2$.
\end{proposition}
\begin{corollary} Any regular multigraph has a strongly spanning maximum $2$-edge-colorable subgraph.
\end{corollary}
\begin{corollary} Any cubic multigraph has a strongly spanning maximum $2$-edge-colorable subgraph.
\end{corollary}

Let us note that the statement of the last corollary for bridgeless cubic multigraphs first appeared in the proof of Theorem 4.1 from \cite{Steffen}. However, an attentive reader probably has already realized that the proof given in \cite{Steffen} is wrong. 

Retaining the notations of \cite{Steffen}, let us, first explain, what is wrong there. The gap is that when the author removes the edges $e_1$ and $e_2$ from a maximum $2$-edge-colorable subgraph $H$ and adds the edges $(v,u_1)$ and $(v,u_2)$ to it to get a new maximum $2$-edge-colorable subgraph $H^{\prime}$, he may leave the other ($\neq u_1$ and $\neq u_2$, respectively) end-vertices isolated, so after this operation one can not conclude that $V(H^{\prime})=V(H)\cup \{v\}$ as it is done there.

\bigskip

Below we offer a generalization of Proposition \ref{almostsp}. Our proof requires the following result of Lov\'asz:
\begin{theorem}(Lov\'asz \cite{LovVal}) If $G$ is a multigraph with $\Delta(G)\leq s+t-1$, then $G$ can be partitioned into two subgraphs $H$ and $L$, such that $\Delta(H)\leq s$ and $\Delta(L)\leq t$.
\end{theorem}

\begin{theorem}\label{Deltadelta} For any multigraph $G$ without isolated vertices $sp(G)\leq \Delta(G)-\delta(G)+2$.
\end{theorem}
\begin{proof}
For a multigraph $G$ take $s=\Delta(G)-\delta(G)+2$ and $t=\delta(G)-1$. Observe that $\Delta(G)= s+t-1$. Apply Lov\'asz's theorem. As a result we have two subgraphs $H$ and $L$, such that $\Delta(H)\leq s$ and $\Delta(L)\leq t$.

Since $\Delta(L)\leq t=\delta(G)-1$, we have $\delta(H)\geq 1$. On the other hand, $\Delta(H)\leq s=\Delta(G)-\delta(G)+2$. Thus $H$ is a $(1,\Delta(G)-\delta(G)+2)$-factor, which proves the theorem. $\square$
\end{proof}Let us note that this bound is tight, since any regular multigraph without a perfect matching achieves it. It can be shown that this bound can be improved by one if $G$ is non-regular (that is, $\Delta(G)\neq \delta(G)$). However, we will not prove this, because below we will prove a significantly better bound for $sp(G)$.

Our next bound is formulated in terms of $\nu(G)$. Its proof requires Theorem 2.1.9 from \cite{Yu}:

\begin{theorem}\label{YuLiuTheorem} \cite{Yu}: Let $b>a\geq 1$. Then a multigraph $G$ has an $[a,b]$-factor, if and only if for all $S\subseteq V(G)$ $\sum_{i=0}^{a-1}(a-i)p_i(G-S)\leq b|S|$, where $p_i(G-S)$ is the number of vertices of degree $i$ in the multigraph $G-S$.
\end{theorem}

\begin{theorem}\label{nuBound} For any multigraph $G$ without isolated vertices  $sp(G)\leq |V(G)|-2\cdot \nu(G)+1$.
\end{theorem}

\begin{proof} By Theorem \ref{YuLiuTheorem} it suffices to show that for each $S\subseteq V(G)$ $p_0(G-S)\leq (|V(G)|-2\cdot \nu(G)+1)|S|$, where $p_0(G-S)$ is the number of isolated vertices of $G-S$. Observe that by Tutte-Berge formula, we have: 
\begin{align*}
p_0(G-S)\leq o(G-S)\leq |S|+(|V(G)|-2\cdot \nu(G))\leq |S|+|S|(|V(G)|-2\cdot \nu(G))=\\
(|V(G)|-2\cdot \nu(G)+1)|S|.
\end{align*}
$\square$
\end{proof}Note that any multigraph with a perfect or a near-perfect matching (a matching missing exactly one vertex) achieves this bound.

Now, we prove the following improvement of Theorem \ref{Deltadelta}:

\begin{theorem}\label{Deltadeltaratio} For any multigraph $G$ without isolated vertices $sp(G)\leq 1+\left\lfloor \frac{\Delta(G)}{\delta(G)}\right\rfloor$. Moreover, if $G$ is non-regular, then $sp(G)\leq \left\lceil  \frac{\Delta(G)}{\delta(G)}\right\rceil$.
\end{theorem}
\begin{proof} Note that since $1+\left\lfloor  \frac{\Delta(G)}{\delta(G)} \right\rfloor> 1$, by Theorem \ref{YuLiuTheorem} it suffices to show that for each $S\subseteq V(G)$ $p_0(G-S)\leq (1+\left\lfloor  \frac{\Delta(G)}{\delta(G)} \right\rfloor)|S|$, where $p_0(G-S)$ is the number of isolated vertices of $G-S$.

Observe that the $p_0(G-S)$ isolated vertices are connected to vertices of $S$, thus
\begin{equation*}
\delta(G)\cdot p_0(G-S)\leq \Delta(G)\cdot |S|,
\end{equation*} which proves the required bound.

For the proof of the second statement, observe that since $G$ is non-regular, then $\left\lceil  \frac{\Delta(G)}{\delta(G)} \right\rceil> 1$, thus Theorem \ref{YuLiuTheorem} is applicable. The rest is the same as above.
$\square$
\end{proof} 

Let us note that there are examples of multigraphs such that the difference between the upper bound offered by Theorem \ref{Deltadeltaratio} and $sp(G)$ is arbitrarily big. To see this, let $H$ be an  $r$-regular multigraph containing a perfect matching $F$. Consider a multigraph $G$ obtained from $H$ by replacing one edge of $F$ by a path of length three. Observe that $G$ contains a perfect matching, hence $sp(G)=1$, however the bound offered by Theorem \ref{Deltadeltaratio} is $\left\lceil  \frac{r}{2}\right\rceil$.

In Theorem \ref{nuBound}, we have shown that an upper bound for $sp(G)$ is provable in terms of the difference between $|V(G)|$ and $\nu(G)$. It is natural to wonder, whether such a bound is possible to prove in terms of the ratio of $|V(G)|$ and $\nu(G)$. The following proposition shows the impossibility of such a bound.

\begin{proposition}For any positive integers $a,b$ there is a tree $G$ with $sp(G)>a(\frac{|V(G)|}{\nu(G)})^b$.
\end{proposition}
\begin{proof} Let $n$ be any positive integer with $n\geq 4$. Set: $k=an^b$ and $x=2k$. Consider the tree $G$ obtained from a path of length $2x$ and the star $K_{1,k}$ by joining the center of the star to one of end-vertices of the path. Observe that: $|V(G)|=3an^b+1$, $\nu(G)=an^b+1$ and $sp(G)=an^b$. Clearly, we have that $sp(G)>a(\frac{|V(G)|}{\nu(G)})^b$.
$\square$
\end{proof}

\end{document}